\documentclass[american,aps,pra,reprint, superscriptaddress]{revtex4-1}
\usepackage{amsmath,amscd,amsthm,amssymb}
\usepackage{mathrsfs}
\usepackage{graphicx}
\usepackage{epsf,epstopdf}
\usepackage[center]{caption2}
\usepackage[all]{xy}

\usepackage[unicode=true,pdfusetitle, bookmarks=true,bookmarksnumbered=false,bookmarksopen=false, breaklinks=false,pdfborder={0 0 0},backref=false,colorlinks=false] {hyperref}
\hypersetup{ colorlinks,linkcolor=myurlcolor,citecolor=myurlcolor,urlcolor=myurlcolor}
\usepackage{graphics,epstopdf,graphicx, amsthm, amsmath, amssymb, times, braket, color, bm}
\usepackage[up]{subfigure}
\usepackage{cleveref}

\definecolor{myurlcolor}{rgb}{0,0,0.7}

\newcommand{\op}[1]{\operatorname{#1}}
\def\be{\begin{equation}}
\def\ee{\end{equation}}
\def\bea{\begin{eqnarray*}}
\def\eea{\end{eqnarray*}}
\def\ot{\otimes}

\theoremstyle{plain}
\newtheorem{thrm}{\protect\theoremname}

\newtheorem{lem}[thrm]{Lemma}
\newtheorem{exam}[thrm]{Example}
\providecommand{\theoremname}{Theorem}

\newcommand{\iinner}[2]{\langle #1 | #2\rangle}
\newcommand{\out}[2]{| #1\rangle\langle #2 |}
\DeclareMathOperator{\trace}{tr}
\newcommand{\ptr}[2]{\trace_{#1}({#2})}
\newcommand{\tr}[1]{\ptr{}{#1}}
\newcommand{\id}{\mathbf{1}}
\newcommand{\ids}[1]{\id^{#1}}
\newcommand*{\myproofname}{Proof}

\makeatother

\def\cF{\mathcal{F}}\def\cH{\mathcal{H}}
\def\cM{\mathcal{M}}
\def\cQ{\mathcal{Q}}
\def\cV{\mathcal{V}}

\theoremstyle{definition}

\theoremstyle{remark}

\begin{document}

 \author{Sunho Kim}
 \email{kimshanhao@126.com}
 \affiliation{Department of Mathematics, Harbin Institute of Technology, Harbin 150001, PR~China}

 \author{Longsuo Li}
 \email{Corresponding author: lilongsuo@126.com}
 \affiliation{Department of Mathematics, Harbin Institute of Technology, Harbin 150001, PR~China}
 
 \author{Asutosh Kumar}
\email{Corresponding author: asutoshk.phys@gmail.com}
\affiliation{P.G. Department of Physics, Gaya College, Magadh University, Rampur, Gaya 823001, India}
% \affiliation{Harish-Chandra Research Institute, HBNI, Chhatnag Road, Jhunsi, Allahabad 211019, India}
% \affiliation{Homi Bhaba National Institute,  Anushaktinagar, Mumbai 400094, India}
%

 \author{Junde Wu}
 \email{Corresponding author: wjd@zju.edu.cn}
 \affiliation{School of Mathematical Sciences, Zhejiang University, Hangzhou 310027, PR~China}

%\title{Local Quantum Fisher Information and Nonclassical Correlations}
\title{Characterizing nonclassical correlations via local quantum Fisher information}
\begin{abstract}
We define two ways of quantifying the quantum correlations based on quantum Fisher information (QFI) in order to study the quantum correlations as a resource in quantum metrology. By investigating the hierarchy of measurement-induced Fisher information introduced in Lu et al. [X. M. Lu, S. Luo, and C. H.
Oh, Phys Rev. A {\bf 86}, 022342 (2012)], %\cite{Lu}, 
we show that the presence of quantum correlation can be confirmed by the difference of the Fisher information induced by the measurements of two hierarchies. In particular, the quantitative quantum correlations based on QFI coincide with the geometric discord for pure quantum states.
\end{abstract}
\maketitle

\section{introduction}

The classification and quantification of correlations are of fundamental importance in science. In the
emerging field of quantum information, various measures of correlations, such as entanglement \cite{Bennett, Vidal, Horodecki}, classical correlations, and quantum correlations \cite{Ollivier, Henderson} have been introduced and studied extensively recently. In particular, the quantum discord was introduced by Ollivier and Zurek \cite{Ollivier} and Henderson and Vedral \cite{Henderson, Vedral} as a measure of quantum correlation beyond entanglement. Another version of quantum discord, namely the geometric discord, was introduced by Daki$\acute{c}\ et\ al$. in Ref. \cite{Dakic}.

Recently, nonclassical correlations (a family of discordlike measures) have been characterized in Ref. \cite{Girolami} within the framework of local quantum uncertainty calculated by utilizing skew information. Nonclassical correlations have also been investigated by other authors from other perspectives. %\cite{Pasquale1, Pasquale2, Pasquale3}.
In particular, ``discriminating length''--a discordlike quantity, has been proposed as a bona fide measure of nonclassical correlations in Refs. \cite{Pasquale1, Pasquale2}. It exploits the quantum Chernov bound which can be seen as a counterpart of the Cramer-Rao bound for discrimination purposes of discrete and continuous parameters.
In the following, we observe the same spirit and 
%propose alternative methods to explore and 
characterize quantum correlations using quantum Fisher information (QFI), a figure of central merit in the quantum estimation theory. The amount of quantum correlation, quantified in terms of the local quantum Fisher information, present in a bipartite mixed probe state can then be used as a resource in quantum metrology.

Information is always encoded in states of physical systems, usually in the form of parameters. In order to extract the encoded information, one has to identify different states, usually via measurements and parameter estimation.
In the classical case, the Fisher information is the central concept in parameter estimation analysis due to the Cram$\acute{e}$r-Rao inequality, which sets a basic lower bound to the variance of any unbiased estimator in terms of the Fisher information \cite{Cramer, Cover}. In the quantum case, the information carried by a physical parameter is usually captured or synthesized by QFI \cite{Helstrom, Holevo, Kay, Genoni}, which is the minimum achievable statistical uncertainty in the estimation of the parameter. Further, if the parameter $\theta$ be related with quantum evolution due to the fixed self-adjoint operator (observable) $H$, \textit{i.e.}, $\rho_\theta = e^{-i\theta H}\rho e^{i\theta H}$, then QFI is independent of the parameter $\theta$ \cite{Luo4} and represents the ``speed of evolution'' of the quantum state (see Ref.  \cite{note1}). 
%
%{\color{red}[The Fisher information is a “measure” of the speed of evolution of the quantum state in the sense that when the speed of evolution of the quantum state is large (i.e., the measurement-induced statistics is weak or fragile), more is the value of the Fisher information, and that QFI becomes zero when it does not evolve.]}

%
In this paper, we address some intriguing questions relating the quantum correlations. First, how does the presence of quantum correlations affect the local quantum Fisher information (lQFI; see Secs.  \uppercase\expandafter{\romannumeral2} and  \uppercase\expandafter{\romannumeral3}A)? Next, for the zero discord bipartite states, which hierarchy of measurement leads to lQFI (as the maximum of Fisher information induced by measurements)? We endeavor to present some answers and ways to investigate these questions.

The paper is organized as follows. In Sec. \uppercase\expandafter{\romannumeral2}, we briefly review quantum Fisher information and quantum discord (both entropic and geometric versions). In Sec. \uppercase\expandafter{\romannumeral3}, we first propose a measure of quantum correlation to study the relationship between quantum correlations and the lQFI driven with a local unitary on party $a$. Next, we recollect the notion of measurement-induced Fisher information (MFI), and show that the MFI and the lQFI on party $b$ are identical for the classical-quantum (CQ) and the classical-classical (CC) states (zero discord bipartite states).  In Sec. \uppercase\expandafter{\romannumeral4}, we show that the quantitative quantum correlations based on QFI coincide with the geometric discord for a pure quantum
state. Finally, Sec. \uppercase\expandafter{\romannumeral5} concludes with a summary.

\section{Quantum Correlations}
\subsection{Quantum Fisher information}
\label{sec:QFI}

Consider we have an $N$-dimensional quantum state $\rho_\theta$ depending on an unknown parameter $\theta$. If we intend to draw out information about $\theta$ from $\rho_\theta$, a set of generalized quantum measurements, namely positive-operator-valued measures (POVMs), $M = \{M_x|M_x \geq 0, \sum_x M_x = \id \}$ should be performed. 
A POVM is said to be ``informationally complete'' if all states can be determined uniquely by the measurement statistics. A
symmetric informationally complete POVM is a special informationally complete POVM that is distinguished by global
symmetry between POVM elements \cite{sic}. What is more, a quantum measurement is ``Fisher-symmetric'' if it provides uniform and maximal information on all parameters that characterize the quantum state of interest \cite{f-sic}. Very recently, Zhu and Hayashi \cite{universal-f-sic} have studied the universally Fisher-symmetric informationally complete measurements.
According to classical statistical theory, the quality of any measurement result can be specified by a form of information called Fisher information (FI) \cite{Cramer} (see also Ref. \cite{note2})
\be\label{eq:FI}
F(\rho_\theta|M) = \frac{1}{4}\int dx p_M(x|\theta)\bigg[\frac{\partial\op{ln}p_M(x|\theta)}{\partial\theta}\bigg]^2,
\ee
where $p_M(x|\theta) = \tr{M_x\rho_\theta}$. 
%{\color{red}Comment: What is the quality of measurement result when FI is zero and non-zero?} 
%It is a measure of information contained in the variable x parametrized by $\theta$. 
Vanishing FI would mean that the measurement-induced statistics is robust in the sense that the probability or the likelihood function is large. 
Optimizing over all possible measurements, we can define the QFI \cite{Braunstein} as
\bea
\cF(\rho_\theta) = \sup_M F(\rho_\theta|M).
\eea
Further, we can rewrite the QFI as
\be\label{eq:QFI}
\cF(\rho_\theta) = \frac{1}{4}\tr{\rho_\theta L^2_\theta},
\ee
by introducing the so-called symmetric logarithmic derivative (SLD) $L_\theta$, satisfying the relation
\bea
\partial_\theta\rho_\theta \equiv \frac{\partial\rho_\theta}{\partial\theta} = \frac{\rho_\theta L_\theta + L_\theta\rho_\theta}{2}.
\eea
It is a Lyapunov matrix equation and a general solution exists. An explicit form of the SLD can be obtained, using the spectral decomposition $\rho_\theta = \sum_{i=1}^N p_i\out{\psi_i}{\psi_i}$, as
\be\label{eq:SLD}
L_\theta = 2\sum_{i,j} \frac{\bra{\psi_j}\partial_\theta\rho_\theta\ket{\psi_i}}{p_i+p_j}\out{\psi_j}{\psi_i},
\ee
where it is understood that the sum is on the indices for which $p_i+p_j \neq 0$. From Eq. (\ref{eq:SLD}) follows the explicit formula for QFI:
\be
\cF(\rho_\theta) = \sum_{i,j} \frac{|\bra{\psi_j}\partial_\theta\rho_\theta\ket{\psi_i}|^2}{2(p_i+p_j)}.
\ee
The above expression of the QFI is further simplified when the quantum state is pure, given by the wave function $|\psi_\theta\rangle$ ($\rho_\theta = \out{\psi_\theta}{\psi_\theta}$). In standard quantum mechanics, it is obtained straight away that $L_\theta = 2\partial_\theta\rho_\theta$ since $\partial_\theta\rho_\theta = \partial_\theta(\rho_\theta)^2 = (\partial_\theta\rho_\theta)\rho_\theta + \rho_\theta(\partial_\theta\rho_\theta)$. It follows that
\be 
\cF(|\psi_\theta\rangle) = \iinner{\partial_\theta\psi}{\partial_\theta\psi} + |\iinner{\partial_\theta\psi}{\psi}|^2.
\ee

In particular, if $\rho_\theta = e^{-i\theta H}\rho e^{i\theta H}$, where $H$ is a fixed observable (Hermitian operator) on the system $\cH$, then $\cF(\rho_\theta)$ is independent of the parameter $\theta$ \cite{Luo4}, and in this circumstance, $\cF(\rho_\theta)$ coincides with $\cF(\rho, H) \equiv \frac{1}{4}\tr{\rho L^2}$, which corresponds to the generalization of Eq. (\ref{eq:QFI}). Here $L$ is determined by
\bea
i[\rho, H] = \frac{1}{2}(L\rho + \rho L),
\eea
where the square bracket denotes the commutator between operators. Moreover, if we know the spectral decomposition $\rho = \sum_{i=1}^N p_i\out{\psi_i}{\psi_i}$, then the QFI can be evaluated as \cite{Luo4, Braunstein}
\be\label{eq:QFI2}
\cF(\rho, H) = \sum_{i,j}\frac{(p_i - p_j)^2}{2(p_i + p_j)}|\bra{\psi_i}H\ket{\psi_j}|^2.
\ee
It satisfies the following information-theoretic property,
\bea\label{eq:Convexity}
\cF\big(\sum_n\lambda_n\rho_n,H\big) \leq \sum_n\lambda_n\cF(\rho_n,H),
\eea
where $\sum_n\lambda_n = 1, \lambda_n\geq0$, and $\rho_n$ are quantum states. The above inequality means that if several different quantum systems are mixed, the information content of the resulting system is not larger than the average information content of the component systems.
Also, for any pure state $\rho$, the QFI reduces to the variance of the observable $H$:
\be\label{eq:QFI of pure}
\cF(\rho, H) = \cV(\rho, H)= \tr{\rho H^2} - \left(\tr{\rho H}\right)^2.
\ee
In general, if $\rho$ is mixed, we have
\bea
0 \leq \cF(\rho, H) \leq \cV(\rho, H).
\eea

\subsection{Quantum discords}

Recall that quantum discord (entropic version) \cite{Ollivier, Henderson} of a bipartite state $\rho^{ab}$ on a system $\cH^a\ot \cH^b$ shared between parties $a$ and $b$ with marginals $\rho^{a} = \ptr{b}{\rho^{ab}}$ and $\rho^{b} = \ptr{a}{\rho^{ab}}$ can be expressed as
$$D_Q(\rho^{ab})\equiv\min_{\Pi^a}\{I(\rho^{ab}) - I[\Pi^a(\rho^{ab})]\}.$$
Here the minimum is over all von Neumann measurements $\Pi^a = \{\Pi^a_k\}_k$ on party $a$, and
$$\Pi^a(\rho^{ab})\equiv \sum_{k}\big(\Pi^a_k\ot \ids{b}\big)\rho^{ab}\big(\Pi^a_k\ot \ids{b}\big),$$
is the resulting state after the measurement. $I(\rho^{ab}) \equiv S(\rho^{a}) + S(\rho^{b}) - S(\rho^{ab})$ is the quantum mutual information, $S(\rho^{a})\equiv -\tr{\rho^{a}\op{ln}\rho^{a}}$ is the von Neumann entropy, and $\ids{b}$ is the identity operator on $\cH^b$. For pure quantum states, it coincides with (the measure of) entanglement. 
Note that quantum discord is a measure of quantum correlation beyond entanglement in the sense that it can be non-zero even for (mixed) separable quantum states, $\rho^{ab} = \sum_{i}p_{i}\rho_i^{a}\otimes \rho_i^{b}$ \cite{Spehner}.  A separable state is classified as (i) CQ if $\rho_i^{a} = |\psi_i\rangle \langle \psi_i|$, (ii) quantum-classical (QC) if $\rho_i^{b} = |\phi_i\rangle \langle \phi_i|$, and (iii) CC if $\rho_i^{a} = |\psi_i\rangle \langle \psi_i|$ and $\rho_i^{b} = |\phi_i\rangle \langle \phi_i|$. 
%Quantum discord vanishes for CQ and CC states, when measurement is performed on party $a$. 
These states are obtained by preparing local mixtures of non-orthogonal states, which cannot be perfectly discriminated by local measurements. In general, quantum discord is difficult to calculate, even for two-qubit states \cite{Luo1, Lang}. Another version of discord, namely geometric discord, introduced by Daki$\acute{c}\ et\ al$. \cite{Dakic}, can be equivalently and conveniently formulated as \cite{Luo2}
$$D_G(\rho^{ab})\equiv\min_{\Pi^a}\|\rho^{ab} - \Pi^a(\rho^{ab})\|^2.$$
Consequently, if $\rho^{ab} = \out{\psi}{\psi}$ is a pure state with the Schmidt decomposition $\ket{\psi} = \sum_i \sqrt{s_i}\ket{\alpha_i}\ot\ket{\beta_i}$, then by the result in Ref. \cite{Luo3} we have
$D_G(\rho^{ab}) = 1-\sum_i s_i^2.$
It is known that both these discords vanish iff the bipartite states are uncorrelated or contain only classical correlations, as the CQ states or the CC states \cite{Dakic}.

%=============================================================================%
\section{Quantum correlations in terms of QFI}
%=============================================================================%

\subsection{In terms of QFI over local observables on party $a$}

Let us consider an $M \times N$ bipartite quantum state $\rho^{ab} = \sum_i p_i\out{\psi_i}{\psi_i}$ on the system $\cH^a\ot \cH^b$. In the case where a single party, say party $a$, is driven with the fixed observable $H_A = H_a \ot \id^b$, $i.e., \ \rho^{ab}_\theta = e^{-i\theta H_A}\rho^{ab} e^{i\theta H_A}$, $\cF(\rho^{ab}, H_A)$ is called the lQFI on party $a$.
Further, let us consider that observable $H_A = \out{\varphi}{\varphi}\ot \id^b$. If $[\rho^{ab}, H_A] = 0$, this means that $\ptr{}{\out{\psi_i}{\psi_i}H_A} = 0$ or $\out{\psi_i}{\psi_i} = \out{\varphi}{\varphi}\ot \out{\beta_i}{\beta_i}$ for all $i$, where $\ket{\beta_i}$s are unit vectors of $\cH^b$, and then we have $\cF(\rho^{ab}, H_A) = 0$. This means that there are no changes in the evolution due to the observable $H_A$, and no information can be obtained through any measurement. Here, we stress on the fact that there exists at least one local observable $H_A$ for which the quantum
speed of evolution is zero for the classical-quantum state $\out{\varphi}{\varphi}\ot \sigma^b$, with $\sigma^b$ being any density operator on $\cH^b$.
Based on this, we define a measure to quantify quantum correlations in terms of the lQFI on party $a$, over all orthonormal bases $\{\ket{\varphi_n}\}_n$ of $\cH^a$, as follows
\bea
\cQ_{a,H}(\rho^{ab}) \equiv \min_{\{H_{A,n}=\out{\varphi_n}{\varphi_n}\ot \id^b\}_n}\sum_n\cF(\rho^{ab},H_{A,n}).
\eea
The $\cQ_{a,H}(\rho^{ab})$ represents the minimal sum of the quantum speed of evolution, when it is driven with the local observables $\{H_{A,n}=\out{\varphi_n}{\varphi_n}\ot \id^b\}_n$ in product subspaces of one-dimensional projective space on party $a$ and $\cH^b$. The following are some desirable properties to be inherited from QFI.

(i) \emph{Non-negativity}. $\cQ_{a,H}(\rho^{ab}) \geq 0$, the equality holds when $\rho^{ab}$ is a CQ or a CC state.

(ii) \emph{Local unitary invariance}. $\cQ_{a,H}(\rho^{ab})$ is invariant under any local unitary operation $\Phi_U$ on parties $a$ and $b$, \emph{i.e.}, $\cQ_{a,H}(\rho^{ab})= \cQ_{a,H}\big(\Phi_U(\rho^{ab})\big)$.

(iii) \emph{Contractivity}. $\cQ_{a,H}(\rho^{ab})$ is contractive under any local completely positive trace preserving (l-CPTP) map $\Phi_l$ on party $b$, \emph{i.e.}, $\cQ_{a,H}(\rho^{ab})\geq \cQ_{a,H}\big(\Phi_l(\rho^{ab})\big)$.

Next, to show that $\cQ_{a,H}(\rho^{ab})$ is a reasonable measure of quantum correlation, we introduce the following theorem.

\begin{thrm}\label{th:mainresult}
$\cQ_{a,H}(\rho^{ab})$ vanishes if and only if $\rho^{ab}$ is a zero discord state, which in turn is equivalent to $\rho^{ab}$ being a CQ state or a CC state.
\end{thrm}
\begin{proof}
The necessary and sufficient condition for vanishing $\cQ_{a,H}(\rho^{ab})$ is that there exists a set of local observables $\{H_{A,n}=\out{\varphi_n}{\varphi_n}\ot \id^b\}_n$, with $\{\ket{\varphi_n}\}_n$ being an orthonormal basis of $\cH^a$, such that $[\rho^{ab}, H_{A,n}] = 0$ for all $n$ (see Ref. \cite{note3}). First, we consider the case when $\rho^{ab}$ is a CQ state, $\rho^{ab} = \sum_i p_i \out{\alpha_i}{\alpha_i}\ot \sigma_i^b$ where $\iinner{\alpha_i}{\alpha_j} = 0$ for $i\neq j.$ (or a CC state, $\rho^{ab} = \sum_i p_i \out{\alpha_i}{\alpha_i}\ot \out{\beta_i}{\beta_i}$ where $\iinner{\alpha_i}{\alpha_j} = \iinner{\beta_i}{\beta_j} = 0$ for $i\neq j$). Then, when $\{H_{A,i} = \out{\alpha_i}{\alpha_i}\ot \id^b, H'_{A,j} = \out{\gamma_j}{\gamma_j}\ot \id^b\}_{i,j}$, with $\{\ket{\gamma_j}\}_j$ being a basis of $\ker{\big\{\ptr{b}{\rho^{ab}}\big\}}$ (surely, $\{\ket{\alpha_i}, \ket{\gamma_j}\}_{i,j}$ is a basis of $\cH^a$), we have $[\rho^{ab}, H_{A,i}] = 0$ and $[\rho^{ab}, H'_{A,j}] = 0$ for any $i$ or $j$. Therefore, $\cQ_{a,H}(\rho^{ab}) = 0.$

Next, if $\cQ_{a,H}(\rho^{ab}) = 0,$ there exists a set of local observables $\{H_{A,n}=\out{\varphi_n}{\varphi_n}\ot \id^b\}_n$, with $\{\ket{\varphi_n}\}_n$ being an orthonormal basis of $\cH^a$, such that $[\rho^{ab}, H_{A,n}] = 0$ for all $n$. Therefore, $\rho^{ab}$ and $\{H_{A,n}\}$ share a common eigenbasis, and this means that $\rho^{ab} = \sum_{n,m}p_{n,m} \out{\varphi_n}{\varphi_n}\ot\out{\beta_{n,m}}{\beta_{n,m}}$ is a spectral decomposition of $\rho^{ab}$, where $\{\ket{\varphi_n}\ot\ket{\beta_{n,m}}\}_m$ are sets of eigenvectors of $H_{A,n}$ for all $n$. Consequently, the state $\rho^{ab} = \sum_{n}p_{n} \out{\varphi_n}{\varphi_n}\ot\sigma^b_n$ is of the classical-quantum form, where $\sigma^b_n = \sum_m \frac{p_{n,m}}{p_n}\out{\beta_{n,m}}{\beta_{n,m}}$ are density operators on $\cH^b$ and $p_n = \sum_m p_{n,m}$ for all $n$.
\end{proof}

This theorem says that quantum correlations are not captured by the measure $\cQ_{a,H}(\rho^{ab})$ only for zero discord states with respect to measurements on party $a$. That is, like quantum discord, $\cQ_{a,H}(\rho^{ab})$ also vanishes for CQ and CC states. It also means that the presence of quantum correlations can be detected by $\cQ_{a,H}(\rho^{ab})$ in separable quantum states.

Furthermore, for any $2\times N$ bipartite quantum states the above theorem is simplified easily as follows:
Let $\rho^{ab} = \sum_i p_i\out{\psi_i}{\psi_i}$ be a bipartite quantum state on the $2\times N$ system $\cH^a\ot \cH^b$. If there exists a local observable $H_A = \out{\varphi}{\varphi}\ot\id^b$, with $\ket{\varphi}$ being a unit vector of $\cH^a$, such that $\cF(\rho^{ab}, H_A) = 0$ (equivalently, $[\rho^{ab},H_A] = 0$), then $\cQ_{a,H}(\rho^{ab})$ vanishes. This is not difficult to prove. Remembering that $\cH^a$ is a two-dimensional system, there exists another unit vector $\ket{\phi}$ such that $\{\ket{\varphi}, \ket{\phi}\}$ is an orthonormal basis of $\cH^a$. Then, for all $i$ we have $\ket{\psi_i} = \ket{\varphi}\ot\ket{\beta_i}$ or $\ket{\psi_i} = \ket{\phi}\ot\ket{\beta_i}$, where $\ket{\beta_i}$'s are unit vectors of $\cH^b$, because $[\rho^{ab},H_A] = 0$ and $[\out{\psi_i}{\psi_i},H_A] \neq 0$ where $\ket{\psi_i} = (a\ket{\varphi}+b\ket{\phi})\ot\ket{\beta_i}$ with $a^2+b^2 = 1$. Consequently, the state $\rho^{ab} = q_{1} \out{\varphi}{\varphi}\ot\sigma^b_1 + q_{2} \out{\phi}{\phi}\ot\sigma^b_2$ is of the classical-quantum form with $q_1+ q_2=1$, where $\sigma^b_i$ are density operators on $\cH^B$ for $i=$ 1 and 2. However, in general, this is not applicable when $M\geq3$. It is illustrated by the following example.

\begin{exam}
Let $\rho^{ab} = \sum_{i=0}^2 p_i\out{\psi_i}{\psi_i}$ be a bipartite quantum state on the $M\times N$-dimensional system $\cH^a\ot \cH^b$ with $M\geq3$, and $\{\ket{\varphi_n}\}^{M-1}_{n=0}$ is an orthonormal basis of $\cH^a$.
Let
\bea
\ket{\psi_0} &=& \ket{\varphi_0}\ot\ket{\beta_0},\\
\ket{\psi_1} &=& \big(a_1\ket{\varphi_1}+a_2\ket{\varphi_2}\big)\ot\ket{\beta_1},\\
\ket{\psi_2} &=& \big(b_1\ket{\varphi_1}+b_2\ket{\varphi_2}\big)\ot\ket{\beta_2},
\eea
where $a_1b_1 + a_2b_2 \neq 0$ and $\ket{\beta_i}$s are unit vectors of $\cH^b$ with $\iinner{\beta_1}{\beta_2} = 0$. Then, even though there exists a local observable $H_A = \out{\varphi_0}{\varphi_0}\ot \id^b$ such that $\cF(\rho^{ab}, H_A) = 0$, $\cQ_{a,H}(\rho^{ab})$ does not vanish.
Note also that the state $\rho^{ab}$ is different from the above form.
\end{exam}

%\subsection{In terms of QFI over local measurements on party $a$}
Similar to the lQFI on party $a$, when the party $b$ is driven with the fixed observable $H_B = \id^a \ot H_b$, $\cF(\rho^{ab}, H_B)$ is called the lQFI on party $b$. Before investigating the lQFI on party $b$, we first introduce some definitions.
Let $\{H_{\mu}\}$ be a complete set of orthonormal observables, that is, $\ptr{}{H_{\mu}H_{\nu}} = \delta_{\mu\nu}$, and $\{H_{\mu}\}$ constitutes a basis for the real Hilbert space of all observables on $\cH$. Then $\sum_\mu\cF(\rho,H_{\mu})$ is independent of the choice of the orthonormal observable basis $\{H_{\mu}\}$ \cite{Li}, that is
\bea
\sum_\mu\cF(\rho,H_{\mu}) = \sum_\nu\cF(\rho,H'_{\nu}),
\eea
where $\{H'_{\nu}\}$ is another orthonormal observable basis.
Similarly, if both $\{H_{b,\mu}\}$ and $\{H'_{b,\nu}\}$ are orthonormal observable bases on $\cH^b$, we may write
\bea
H'_{B,\nu} &=& \id^a \ot H'_{b,\nu} = \sum_\mu c_{\nu\mu}\id^a \ot H_{b,\mu}\\
&=& \sum_\mu c_{\nu\mu}H_{B,\mu}; \quad \nu = 1,2,\ldots,N^2,
\eea
with $\{c_{\nu\mu}\}_{\nu,\mu}$ being an $N^2\times N^2$ real orthonormal matrix, that is,
\bea
\sum_\nu c_{\nu\mu}c_{\nu r} = \delta_{\mu r}, \quad \mu, r = 1,2,\ldots,N^2.
\eea
Then,
\bea
&&\sum_\nu\cF(\rho^{ab},H'_{B,\nu})\\
&=& \sum_\nu\sum_{i,j}\frac{(p_i - p_j)^2}{2(p_i + p_j)}|\bra{\psi_i}\sum_\mu c_{\nu\mu}H_{B,\mu}\ket{\psi_j}|^2\\
&=&\sum_\nu\sum_{i,j}\frac{(p_i - p_j)^2}{2(p_i + p_j)}\bra{\psi_i}\sum_\mu c_{\nu\mu}H_{B,\mu}\ket{\psi_j}\bra{\psi_j}\sum_r c_{\nu r}H_{B,r}\ket{\psi_i}\\
&=&\sum_{\mu r}\big(\sum_\nu c_{\nu\mu}c_{\nu r}\big)\sum_{i,j}\frac{(p_i - p_j)^2}{2(p_i + p_j)}\bra{\psi_i}H_{B,\mu}\ket{\psi_j}\bra{\psi_j}H_{B,r}\ket{\psi_i}\\
&=&\sum_{\mu}\sum_{i,j}\frac{(p_i - p_j)^2}{2(p_i + p_j)}|\bra{\psi_i}H_{B,\mu}\ket{\psi_j}|^2\\
&=&\sum_{\mu}\cF(\rho^{ab},H_{B,\mu}),
\eea
where $\rho^{ab} = \sum_i p_i\out{\psi_i}{\psi_i}$. Thus, $\sum_{\mu}\cF(\rho^{ab},H_{B,\mu})$ is independent of the choice of the orthonormal local observable basis $\{H_{b,\mu}\}$ on $\cH^b$, which means that we can evaluate it using any orthonormal observable basis on $\cH^b$.

\subsection{In terms of QFI over local measurements on party $a$}

Let $\cM^a$ and $\cM^b$ be sets of all local measurements on parties $a$ and $b$, respectively. And $\cM^{a\rightarrow b}$ is the set of joint POVMs that party $b$ performs after party $a$, conditioned on the outcomes of party $a$. Then the MFI for $\cM^{a\rightarrow b}$ is defined as
\bea
\cF(\rho^{ab}, H_B |\cM^{a\rightarrow b}) \equiv \max_{\Lambda\in\cM^{a\rightarrow b}}F(\rho^{ab}_\theta |\Lambda),
\eea
where $\rho^{ab}_\theta = e^{-i\theta H_B}\rho^{ab} e^{i\theta H_B}$ (the MFI can be defined as above, without any logical loss, because the maximum of FI over product measurements $\Lambda\in\cM^{a\rightarrow b}$ is independent of the parameter $\theta$ associated with local observables).  Obviously, $\cF(\rho^{ab}, H_B |\cM^{a\rightarrow b}) \leq \cF(\rho^{ab},H_B)$, because $\cM^{a\rightarrow b}\subseteq \cM^{ab}$ where $\cM^{ab}$ is the entire set of all POVMs on the composite system.
Furthermore, we can infer that the MFI is the maximal FI captured by using only classical correlations with respect to measurements on party $a$, because it is induced by a product measurement that party $b$ performs after party $a$, conditioned on the outcomes of party $a$.
Meanwhile, when the local measurement on party $a$ is fixed to the von Neumann measurement $\Pi^a$, we define $\cM^{\Pi^a\rightarrow b}$ as the set of product measurements that party $b$ performs after party $a$, conditioned on the outcomes of von Neumann measurements $\Pi^a$ of party $a$. Then it can be expressed as \cite{Lu}
\bea
\cF(\rho^{ab}, H_B |\cM^{\Pi^a\rightarrow b})
&=& \cF\big(\Phi_{\Pi^a}(e^{-i\theta H_B}\rho^{ab} e^{i\theta H_B})\big)\\
&=& \cF\big(e^{-i\theta H_B}\Phi_{\Pi^a}(\rho^{ab}) e^{i\theta H_B}\big)\\
&=& \cF(\Phi_{\Pi^a}(\rho^{ab}), H_B)\\
&=& \sum_n p^a(n)\cF(\rho^{b|n}, H_b),
\eea
where $\Phi_{\Pi^a}(\sigma) \equiv \sum_n(\Pi^a_n\ot \id^b) \sigma(\Pi^a_n\ot \id^b)$, and
\be\label{eq:state related to measurement}
\rho^{b|n} \equiv \frac{\ptr{a}{\Pi^a_n\ot \id^b\rho^{ab}}}{p^a(n)},
\ee
with $p^a(n)\equiv\tr{\Pi^a_n\ot \id^b\rho^{ab}}$.
%{\color{red}Comment: (i) How are lQFI and MFI different from entropic quantum discord? (ii) Is there any connection/relation between FI and entropy? (iii) Can we give a schematic or tree diagram of the work/results in this paper?} 

Now we compare these two pieces of FI for the CQ states, $i.e., \rho^{ab} = \sum_{i} p_{i} \out{\alpha_i}{\alpha_i}\ot\rho^{b|i}$. First, lQFI is calculated as
\bea
&&\cF(\rho^{ab},H_B)\\
&=& \sum_{i,j,m,n}\frac{(p_{i}p_{i,n} - p_{j}p_{j,m})^2}{2(p_{i}p_{i,n} + p_{j}p_{j,m})}|\iinner{\alpha_i}{\alpha_j}\bra{\beta_{i,n}}H_b\ket{\beta_{j,m}}|^2\\
&=& \sum_{i,m,n}p_{i}\frac{(p_{i,n} - p_{i,m})^2}{2(p_{i,n} + p_{i,m})}|\bra{\beta_{i,n}}H_b\ket{\beta_{i,m}}|^2\\
&=& \sum_{i}p_{i}\cF(\rho^{b|i},H_b),
\eea
where $\rho^{b|i} = \sum_n p_{i,n}\out{\beta_{i,n}}{\beta_{i,n}}$ for all $i$.
Interestingly, when $\Pi^a = \{\out{\alpha_i}{\alpha_i}\}_i$, we also have
\bea
\cF(\rho^{ab},H_{B}|\cM^{\Pi^a\rightarrow b}) = \sum_{i}p_{i}\cF(\rho^{b|i},H_b).
\eea
Therefore, we obtain that
\bea
\cF(\rho^{ab}, H_B |\cM^{\Pi^a\rightarrow b}) = \cF(\rho^{ab},H_B).
\eea
This means that the lQFI $\cF(\rho^{ab},H_B)$ is read by a joint measurement of parties $a$ and $b$ in $\cM^{\Pi^a\rightarrow b}$ for the CQ states (or CC states) $\rho^{ab}$, and it is applied for all local observables on party $b$. In addition, the party $a$ of the joint measurement that leads lQFI remains unchanged for any local observable on party $b$.

Based on this result, we define the following as an indicator of quantum correlation:
\bea
&&\cQ_{a,\Pi}(\rho^{ab})\\
&\equiv& \sum_{\mu}\cF(\rho^{ab},H_{B,\mu}) - \max_{\Pi^a}\sum_{\mu}\cF(\rho^{ab},H_{B,\mu}|\cM^{\Pi^a\rightarrow b}),
\eea
where the maximum is over all local von Neumann measurements $\Pi^a = \{\Pi^a_n = \out{n}{n}\}$ on party $a$.
Consequently, if $\rho^{ab}$ is any zero discord state, then $\cQ_{a,\Pi}(\rho^{ab}) = 0$.
%\bea
%\cQ_{a,\Pi}(\rho^{ab}) = 0.
%\eea
We prove that the opposite also holds true.

\begin{thrm}\label{th:mainresult2}
$\cQ_{a,\Pi}(\rho^{ab})$ vanishes if and only if $\rho^{ab}$ is a zero discord state.
\end{thrm}
\begin{proof}
We only need to prove that $\rho^{ab}$ is a CQ state (or a CC state) if $\cQ_{a,\Pi}(\rho^{ab}) = 0$.
In general, when the state $\rho^{ab}$ is not a CQ state, we have the following relations:
\bea
\sum_\mu \cF(\rho^{ab},H_{B,\mu})
&\geq& \sum_\mu \max_{\Pi^a}\cF(\rho^{ab},H_{B,\mu}|\cM^{\Pi^a\rightarrow b})\\
&\geq& \max_{\Pi^a}\sum_{\mu}\cF(\rho^{ab},H_{B,\mu}|\cM^{\Pi^a\rightarrow b}).
\eea
If the local von Neumann measurement on party $a$ that induces the maximum of MFI for at least one local observable on party $b$ is not the same as the others, then
\bea
&&\sum_\mu \max_{\Pi^a}\cF(\rho^{ab},H_{B,\mu}|\cM^{\Pi^a\rightarrow b})\quad\\
&>& \max_{\Pi^a}\sum_{\mu}\cF(\rho^{ab},H_{B,\mu}|\cM^{\Pi^a\rightarrow b}).
\eea
This means that $\cQ_{a,\Pi}(\rho^{ab}) > 0$.

On the other hand, if the local measurements on party $a$ that induce the maximum of MFI for all local observables on party $b$ are identical (let the local measurement be $\Pi^a$),
then there is at least one observable $H_B = \id^a\ot H_b$, such that
\bea
\cF(\rho^{ab},H_{B}) &>& \cF(\rho^{ab},H_{B}|\cM^{\Pi^a\rightarrow b})\\
&=& \sum_n p^a(n)\cF(\rho^{b|n}, H_b)
\eea
with the states $\rho^{b|n}$ of Eq. (\ref{eq:state related to measurement}) related to the $\Pi^a$.
This means that $\cQ_{a,\Pi}(\rho^{ab}) > 0$.

To prove this, let us first assume that
\bea
\cF(\rho^{ab},H_{B}) &=& \cF(\rho^{ab},H_{B}|\cM^{\Pi^a\rightarrow b})\\
&=& p_1\cF(\rho^{b|1},H_{b}) + p_2\cF(\rho^{b|2},H_{b})
\eea
for any local observable $H_{B} = \id^a \ot H_{b}$, and let two observables $H_{b,1}$ and $H_{b,2}$, satisfy $[\rho^{b|1}, H_{b,1}] = 0$ and $[\rho^{b|2}, H_{b,2}] = 0$, respectively. Then,
\bea
\cF(\rho^{ab},H_{B,1}) = p_2\cF(\rho^{b|2},H_{b,1})
\eea
and
\bea
\cF(\rho^{ab},H_{B,2}) = p_1\cF(\rho^{b|1},H_{b,2}).
\eea
%{\color{red}Comment: Are the superscripts and subscripts on the rhs of above expressions correct? Couldn't follow!} When [\rho^1, H_1]=0, it implies that F(\rho^1, H_1)=0. Therefore, by our assumption, F(\rho, H_1)=p_2 F(...).
Also, when $H'_b = (H_{b,1} + H_{b,2})/2$,
\bea
\cF(\rho^{ab},H'_{B}|\cM^{\Pi^a\rightarrow b}) =  \frac{p_1\cF(\rho^{b|1},H_{b}) + p_2\cF(\rho^{b|2},H_{b})}{4},
\eea
and
\bea
\cF(\rho^{ab},H'_{B}) &=& \frac{\cF(\rho^{ab},H_{B,1}) + \cF(\rho^{ab},H_{B,2})}{4}\\
&+& \sum_{i,j}\frac{(p_i-p_j)^2}{p_i+p_j}\bra{\psi_i}H_{B,1}\ket{\psi_j}\bra{\psi_j}H_{B,2}\ket{\psi_i}\\
&=& \frac{p_1\cF(\rho^{b|1},H_{b}) + p_2\cF(\rho^{b|2},H_{b})}{4}\\
&+& \sum_{i,j}\frac{(p_i-p_j)^2}{p_i+p_j}\bra{\psi_i}H_{B,1}\ket{\psi_j}\bra{\psi_j}H_{B,2}\ket{\psi_i},
\eea
where $\rho^{ab} = \sum_i p_i\out{\psi_i}{\psi_i}$ and $H'_{B} = \id^a \ot H'_b$.
But, only when $\rho^{ab}$ is a CQ state (\emph{i.e.}, $\rho^{ab} = p_1\out{1}{1}\ot\rho^{b|1} + p_2\out{2}{2}\ot\rho^{b|2}$), we have $$\sum_{i,j}\frac{(p_i-p_j)^2}{p_i+p_j}\bra{\psi_i}H_{B,1}\ket{\psi_j}\bra{\psi_j}H_{B,2}\ket{\psi_i} = 0,$$ for local observables $H_{B,1} = \id^a \ot H_{b,1}$ and $H_{B,2} = \id^a \ot H_{b,2}$ satisfying $[\rho^{b|1}, H_{b,1}] = 0$ and $[\rho^{b|2}, H_{b,2}] = 0$ respectively. This means that $\cF(\rho^{ab},H'_{B}) > \cF(\rho^{ab},H'_{B}|\cM^{\Pi^a\rightarrow b})$ when $\rho^{ab}$ is not a CQ state.

For generalization, we assume that\\
\parbox{8.3cm}{
\begin{eqnarray*}
\cF(\rho^{ab},H_{B}) &=& \cF(\rho^{ab},H_{B}|\cM^{\Pi^a\rightarrow b})\\
&=& \sum_n p^a(n)\cF(\rho^{b|n},H_{b}),
\end{eqnarray*}}\hfill
\parbox{.3cm}{\begin{eqnarray}\label{eq:1}\end{eqnarray}}\\
for any local observable $H_{B} = \id^a \ot H_{b}$. Let $\Psi_{k,l}$ be local operations on party $a$, defined as
\bea
\Psi_{k,l}(\sigma) &=& \sum_{n\neq k,l} (\Pi^a_n\ot \id^b) \sigma(\Pi^a_n\ot \id^b)\\
&& + \left((\Pi^a_k+\Pi^a_l)\ot \id^b\right) \sigma\left((\Pi^a_k+\Pi^a_l)\ot \id^b\right),
\eea
for $k\neq l$ and any state $\sigma$. Then, for any local observable $H_{B} = \id^a \ot H_b$ and any $k\neq l$, we have
\bea
&&\cF(\rho^{ab},H_{B})\\
&\geq& \cF(\Psi_{k,l}(\rho^{ab}),H_{B})\\
&=& \sum_{n\neq k,l} p^a(n)\cF(\rho^{b|n},H_{b}) + (p^a(k) + p^a(l))\cF(\rho^{ab}_{k,l},H_{B})\\
&\geq& \sum_{n} p^a(n)\cF(\rho^{b|n},H_{b}),
\eea
where $$\rho^{ab}_{k,l} = \frac{\left((\Pi^a_k+\Pi^a_l)\ot \id^b\right) \rho^{ab}\left((\Pi^a_k+\Pi^a_l)\ot \id^b\right)}{p^a(k)
+ p^a(l)}.$$
The first inequality follows by the property of lQFI \cite{Luo5}, and the third inequality follows by
\bea
\cF(\rho^{ab}_{k,l},H_{B}) &\geq& \cF(\rho^{ab}_{k,l},H_{B}|\cM^{\Pi^a\rightarrow b})\\
&=& \frac{p^a(k)\cF(\rho^{b|k},H_{b}) + p^a(l)\cF(\rho^{b|l},H_{b})}{p^a(k) + p^a(l)}.
\eea
So, from Eq.(\ref{eq:1}), we have
\bea
\cF(\rho^{ab}_{k,l},H_{B}) = \frac{p^a(k)\cF(\rho^{b|k},H_{b}) + p^a(l)\cF(\rho^{b|l},H_{b})}{p^a(k)
+ p^a(l)},
\eea
for all $k\neq l.$ Therefore, $\rho^{ab}_{k,l}$ are the CQ states for all $k\neq l$, implying that $\rho^{ab}$ is a CQ state.
\end{proof}

Generally, when $\rho^{ab}$ is not a CQ state (or a CC state), the local measurements on party $a$ that induce the maximum of MFI for the local observables on party $b$ are not all the same. For example, if $\rho^{ab}$ is pure but not a CC state, with the Schmidt decomposition $\ket{\psi} = \sum_{i=1}^L \sqrt{s_i}\ket{\alpha_i}\ot\ket{\beta_i}$ where $1<L\leq \min\{M,N\}$, then the lQFI can be evaluated, with $H_B = \id^a\ot H_b$, as
\\
\parbox{8.3cm}{
\begin{eqnarray*}
&& \cF(\rho^{ab},H_B)\\
&=& \sum_i^L s_i \bra{\beta_i}H_b^2\ket{\beta_i} - \big(\sum_i^L s_i \bra{\beta_i}H_b\ket{\beta_i}\big)^2.
\end{eqnarray*}}\hfill
\parbox{.3cm}{\begin{eqnarray}\label{eq:lQFI in pure}\end{eqnarray}}

On the other hand, in the evaluation of MFI, the state $\rho^{b|n}$ in Eq. (\ref{eq:state related to measurement}) is pure for any $n$, as seen below:
\bea
\rho^{b|n} &=& \frac{\ptr{a}{\Pi^a_n\ot\id^b\rho^{ab}}}{p^a(n)}\\
&=& \frac{\sum^L_{i,j}\sqrt{s_is_j}\iinner{\alpha_j}{n}\iinner{n}{\alpha_i}\out{\beta_i}{\beta_j}}{p^a(n)}\\
&=& \frac{(\sum^L_i\sqrt{s_i}\iinner{n}{\alpha_i}\ket{\beta_i})(\sum^L_i\sqrt{s_i}\iinner{n}{\alpha_i}\ket{\beta_i})^{\ast}}{p^a(n)},
\eea
where $\Pi^a_n = \out{n}{n}$ and $p^a(n) = \sum^L_is_i|\iinner{\alpha_i}{n}|^2$. Therefore, the MFI can be evaluated as\\
\parbox{8.3cm}{
\begin{eqnarray*}
& &\sum_np^a(n)\cF(\rho^{b|n},H_b)\\
&=&
\sum_np^a(n)\bigg\{\frac{\sum^L_{i,j}\sqrt{s_is_j}\iinner{\alpha_j}{n}\iinner{n}{\alpha_i}\bra{\beta_j}H^2_b\ket{\beta_i}}{p^a(n)}\\
&&-\bigg(\frac{\sum^L_{i,j}\sqrt{s_is_j}\iinner{\alpha_j}{n}\iinner{n}{\alpha_i}\bra{\beta_j}H_b\ket{\beta_i}}{p^a(n)}\bigg)^2\bigg\}\\
&=& \sum^L_is_i\bra{\beta_i}H^2_b\ket{\beta_i}\\
&&-\sum_n\frac{\big(\sum^L_{i,j}\sqrt{s_is_j}\iinner{\alpha_j}{n}\iinner{n}{\alpha_i}\bra{\beta_j}H_b\ket{\beta_i}\big)^2}{p^a(n)}.
\end{eqnarray*}}\hfill
\parbox{.3cm}{\begin{eqnarray}\label{eq:MFI in pure}\end{eqnarray}}\\

Since $\cF(\rho^{ab},H_{B}) \geq \sum_{n} p^a(n)\cF(\rho^{b|n},H_{b})$, we obtain the following inequality:
\bea
(1) &\equiv& \big(\sum_i^L s_i \bra{\beta_i}H_b\ket{\beta_i}\big)^2\\
&\leq& \sum_n\frac{\big(\sum^L_{i,j}\sqrt{s_is_j}\iinner{\alpha_j}{n}\iinner{n}{\alpha_i}\bra{\beta_j}H_b\ket{\beta_i}\big)^2}{\sum^L_is_i|\iinner{\alpha_i}{n}|^2} \equiv (2).
\eea
Next, we can see that the local measurements on party $a$ that induce the maximum of MFI are different in the following two situations.

(i) If $[\beta_i, H_b] = 0$ for all $i$, then, by using the measurement $\Pi^a$ that satisfies $|\iinner{\alpha_i}{n}|^2 = \frac{1}{M}$ for all $i$ and $n$, we have
\bea
(2) &=& \sum_n\frac{\big((1/M)\sum^L_{i}s_i\bra{\beta_i}H_b\ket{\beta_i}\big)^2}{(1/M)\sum^L_is_i}\\
&=& \big(\sum_i^L s_i \bra{\beta_i}H_b\ket{\beta_i}\big)^2 = (1).
\eea

(ii) If there exists a real number $r\neq0$ such that $\bra{\beta_i}H_b\ket{\beta_i} = r$ for all $i$, when $\bra{\beta_j}H_b\ket{\beta_i} = 0$ or $\iinner{\alpha_j}{n}\iinner{n}{\alpha_i} = 0$ for all $n$ and $i\neq j$, then
\bea
(2) = \sum_n\frac{\big(r\sum^L_{i}s_i|\iinner{\alpha_i}{n}|^2\big)^2}{\sum^L_is_i|\iinner{\alpha_i}{n}|^2} = r^2 = (1).
\eea

Consequently, the quantity $\cQ_{a,\Pi}$ has the following properties:

(i)\label{1}  $\cQ_{a,\Pi}(\rho^{ab}) \geq 0$, the equality holds only when $\rho^{ab}$ is a CQ or a CC state.

(ii)\label{2} $\cQ_{a,\Pi}(\rho^{ab})$ is invariant under any local unitary operation $\Phi_U$ on parties $a$ and $b$, \emph{i.e.}, $\cQ_{a,\Pi}(\rho^{ab})= \cQ_{a,\Pi}\big(\Phi_U(\rho^{ab})\big)$.

Property (ii) follows readily from the corresponding property of QFI.

Besides, we expect that $\cQ_{a,\Pi}(\rho^{ab})$ is contractive under any lCPTP maps $\Phi_b$ on party $b$, \emph{i.e.}, $\cQ_{a,\Pi}(\rho^{ab})\geq \cQ_{a,\Pi}\big(\Phi_b(\rho^{ab})\big)$, then the quantity $\cQ_{a,\Pi}$ is not only an indicator but also a full-fledged measure of quantum correlations. However, we do not have a proof at present and leave it as an open problem.

%=============================================================================%
\section{Quantification of quantum correlations in pure states}
%=============================================================================%

In this section, we evaluate the quantities $\cQ_{a,H}$ and $\cQ_{a,\Pi}$ for pure quantum states.
Already, nonclassical correlations have been computed for arbitrary $2\times N$ bipartite pure states $\rho^{ab}$ \cite{Girolami} utilizing skew information. Here, we generalize it by calculating the quantified quantum correlations we defined for arbitrary $M\times N$ bipartite pure states $\rho^{ab}$.
Let us consider that $\rho^{ab} = \out{\psi}{\psi}$ is a pure state, with the Schmidt decomposition $\ket{\psi} = \sum_i \sqrt{s_i}\ket{\alpha_i}\ot\ket{\beta_i}$. Then the quantities $\cQ_{a,H}(\rho^{ab})$ and $\cQ_{a,\Pi}(\rho^{ab})$ coincide with the geometric discord $D_G(\rho^{ab})$, that is
\bea
\cQ_{a,H}(\rho^{ab}) = \cQ_{a,\Pi}(\rho^{ab})= D_G(\rho^{ab}) = 1-\sum_i s_i^2.
\eea

To show that $\cQ_{a,H}(\rho^{ab})$ equals geometric discord for pure quantum states, we need the following lemma.
\begin{lem}\label{lem:lem-1}\cite{Watrous}
Let $\cH$ be a Hilbert space, $H$ be a Hermitian operator on $\cH$, and $\{\ket{x_n}\}_n$ an orthonormal basis of $\cH$. Define a vector $u$ as $u(n) = \bra{x_n}H\ket{x_n}$, then $u \prec v$ (called $v$  majorizes $u$, if there exists a doubly stochastic operator $A$ such that $v = Au$), where $v$ is a vector of eigenvalues of $H$. Furthermore, we can easily show that $\sum_iu_i^2 \leq \sum_j v_j^2$ when $u\prec v.$
\end{lem}
In general, when QFI  is driven with local observables $\{H_{A,n}=\out{\varphi_n}{\varphi_n}\ot \id^b\}_n$, where $\{\ket{\varphi_n}\}_n$ is an orthonormal basis of $\cH^a$, we have
\bea
&&\sum_n\cF(\rho^{ab},H_{A,n})\\
&=& \sum_n \bigg[\tr{\rho^{ab}H_{A,n}^2} - \{\tr{\rho^{ab}H_{A,n}}\}^2\bigg]\\
&=& \sum_n\bigg\{\bra{\psi}H_{A,n}^2\ket{\psi} - (\bra{\psi}H_{A,n}\ket{\psi})^2\bigg\}\\
&=& \sum_n\bigg\{\sum_{i}s_i|\iinner{\alpha_i}{\varphi_n}|^2 - \big(\sum_{i}s_i|\iinner{\alpha_i}{\varphi_n}|^2\big)^2\bigg\}\\
&=& 1-\sum_n \big(\sum_{i}s_i|\iinner{\alpha_i}{\varphi_n}|^2\big)^2\\
&\geq& 1-\sum_is_i^2.
\eea
The first equality follows from Eq. (\ref{eq:QFI of pure}) for pure quantum states, and the sixth inequality follows using $$\sum_{i}s_i|\iinner{\alpha_i}{\varphi_n}|^2 = \bra{\varphi_n}\big(\sum_{i}s_i\out{\alpha_i}{\alpha_i}\big)\ket{\varphi_n}$$
and Lemma \ref{lem:lem-1}.
The minimum value of $\cQ_{a,H}(\rho^{ab})$ is then established identically that $\cQ_{a,H}(\rho^{ab})\geq 1-\sum_is_i^2$.

In particular, we may take
$$\{H_{A,i}=\out{\alpha_i}{\alpha_i}\ot \id^b, H'_{A,j}=\out{\gamma_j}{\gamma_j}\ot \id^b\}_{i,j},$$
where $\{\ket{\gamma_j}\}_j$ is a basis of $\ker{\big\{\ptr{b}{\out{\psi}{\psi}}\big\}}$. Then, by straightforward calculations, we obtain
\bea
\sum_i\cF(\rho^{ab},H_{A,i}) + \sum_j\cF(\rho^{ab},H'_{A,j}) = 1-\sum_is_i^2.
\eea
Hence, $\cQ_{a,H}(\rho^{ab})= 1-\sum_is_i^2$.

Next, for $\cQ_{a,\Pi}(\rho^{ab})$, we have the following relations from Eqs. (\ref{eq:lQFI in pure}) and (\ref{eq:MFI in pure}),
\bea
&&\cQ_{a,\Pi}(\rho^{ab})\\
&=& \sum_\mu\cF(\rho^{ab},H_{B,\mu}) - \max_{\Pi^a}\sum_\mu \big(\sum_np^a(n)\cF(\rho^{b|n},H_{b,\mu}) \big)\\
&=&\min_{\Pi^a}\sum_\mu \left(\sum_n\frac{\big(\sum_{i,j}\sqrt{s_is_j}\iinner{\alpha_j}{n}\iinner{n}{\alpha_i}\bra{\beta_j}H_{b,\mu}\ket{\beta_i}\big)^2}{\sum_is_i|\iinner{\alpha_i}{n}|^2} \right)\\
&& - \sum_\mu\big(\sum_i s_i \bra{\beta_i}H_{b,\mu}\ket{\beta_i}\big)^2,
\eea
where $\Pi^a_n = \out{n}{n}$ and $\{H_{b,\mu}\}$ is any orthonormal observable basis on $\cH^b$.
In particular, we may take
\bea
\{H_{b,\mu}\} = \{E_k, E_{k,l}^+, E_{k,l}^-\}
\eea
with
\bea
E_k &=& \out{\beta_k}{\beta_k}, \qquad k = 1,2,\ldots,N,\\
E_{k,l}^+ &=& \frac{1}{\sqrt{2}}(\out{\beta_k}{\beta_l} + \out{\beta_l}{\beta_k}), \ k<l,\ k,l = 1,2,\ldots,N,\\
E_{k,l}^- &=& \frac{i}{\sqrt{2}}(\out{\beta_k}{\beta_l} - \out{\beta_l}{\beta_k}), \ k<l,\ k,l = 1,2,\ldots,N.
\eea
Then, by straightforward calculations, we have
\begin{widetext}
\bea
\big(\sum_i s_i \bra{\beta_i}E_k\ket{\beta_i}\big)^2 &=& s_k^2,\\
\big(\sum_i s_i \bra{\beta_i}E_{k,l}^+\ket{\beta_i}\big)^2 &=& \big(\sum_i s_i \bra{\beta_i}E_{k,l}^-\ket{\beta_i}\big)^2 = 0,\\
\sum_n\frac{\big(\sum_{i,j}\sqrt{s_is_j}\iinner{\alpha_j}{n}\iinner{n}{\alpha_i}\bra{\beta_j}E_k\ket{\beta_i}\big)^2}{\sum_is_i|\iinner{\alpha_i}{n}|^2} &=& \sum_n\frac{\big(s_k|\iinner{\alpha_k}{n}|^2\big)^2}{\sum_is_i|\iinner{\alpha_i}{n}|^2},\\
\sum_n\frac{\big(\sum_{i,j}\sqrt{s_is_j}\iinner{\alpha_j}{n}\iinner{n}{\alpha_i}\bra{\beta_j}E_{k,l}^+\ket{\beta_i}\big)^2}{\sum_is_i|\iinner{\alpha_i}{n}|^2} &=& \sum_n\frac{s_ks_l\big(\iinner{\alpha_k}{n}\iinner{n}{\alpha_l} + \iinner{\alpha_l}{n}\iinner{n}{\alpha_k}\big)^2}{2\sum_is_i|\iinner{\alpha_i}{n}|^2},\\
\sum_n\frac{\big(\sum_{i,j}\sqrt{s_is_j}\iinner{\alpha_j}{n}\iinner{n}{\alpha_i}\bra{\beta_j}E_{k,l}^-\ket{\beta_i}\big)^2}{\sum_is_i|\iinner{\alpha_i}{n}|^2} &=& -\sum_n\frac{s_ks_l\big(\iinner{\alpha_k}{n}\iinner{n}{\alpha_l} - \iinner{\alpha_l}{n}\iinner{n}{\alpha_k}\big)^2}{2\sum_is_i|\iinner{\alpha_i}{n}|^2},
\eea
and consequently,
\bea
\cQ_{a,\Pi}(\rho^{ab}) &=& \min_{\Pi^a}\sum_n \left(\sum_k \frac{\big(s_k|\iinner{\alpha_k}{n}|^2\big)^2}{\sum_is_i|\iinner{\alpha_i}{n}|^2}
+ 2\sum_{k<l} \frac{ s_ks_l|\iinner{\alpha_k}{n}|^2|\iinner{\alpha_l}{n}|^2}{\sum_is_i|\iinner{\alpha_i}{n}|^2} \right) - \sum_k s_k^2\\
&=& \min_{\Pi^a}\sum_n \left(\frac{\big(\sum_k s_k|\iinner{\alpha_k}{n}|^2\big)^2}{\sum_is_i|\iinner{\alpha_i}{n}|^2} \right) - \sum_k s_k^2\\
&=& 1 - \sum_k s_k^2.
\eea
\end{widetext}
Also, we can see that $\cQ_{a,\Pi}(\rho^{ab})$ is independent of the choice of local measurement $\Pi^a$ on party $a$ for any pure state $\rho^{ab}$. This means that we can capture it in terms of any local von Neumann measurement on party $a$.

In particular, on system $\cH^a\ot \cH^b$ with $\dim\cH^a = \dim\cH^b = M$, the two quantities achieve the maximum value $1-1/M$ for any maximally entangled pure state.

%=============================================================================%
\section{Conclusions}
%=============================================================================%

In this paper, we have related quantum Fisher information, a central figure of merit in quantum estimation theory, to the analysis of quantum correlations. In particular, we studied the characterization of quantum correlations from the perspective of the lQFI on parties $a$ and $b$. To do this, we surveyed the lQFI in two ways (via local observables and via local measurements). We first proposed a measure of quantum correlations based on the lQFI for local observables, and showed that it qualifies for a reasonable measure of quantum correlations. This means that we can quantify the quantum correlations through the measure utilizing QFI.

Also, by investigating the hierarchy of measurement-induced local Fisher information, we have shown that the lQFI, for zero discord states only, can be induced by performing a joint measurement, the parts of which act independently on their respective local systems, and is associated with classical correlations only. This means that the presence of quantum correlations leads to the instantaneous variation of evolution that cannot be confirmed by any joint measurement in the bipartite quantum state.
Based on this, we define an indicator of quantum correlations by the difference between the lQFI and the maximum of measurement-induced Fisher information over all local von Neumann measurements. The indicator has still to satisfy the condition that it is contractive under any local map on party $b$ in order to become the full-fledged measure of quantum correlations, but we leave this as an open problem.

Finally, we showed that these two quantities coincide with the geometric discord for any pure quantum state.

\begin{acknowledgments}
This project is supported by the National Natural Science Foundation of China (Grants No. 11171301 and No. 11571307).
\end{acknowledgments}

%\bibliographystyle{apsrev4-1}
% \bibliography{coh-dis-lit-1}

\begin{thebibliography}{99}
\bibitem{note1} The Fisher information is a ``measure'' of the speed of evolution of the quantum state,  $\rho_\theta = e^{-i\theta H}\rho e^{i\theta H}$, in the sense that when the evolution of the quantum state is fast (i.e., the measurement-induced statistics is weak or fragile) the value of the Fisher information is greater and that the QFI becomes zero when it does not evolve.
\bibitem{note2} Note that in the general (classical) definition of Fisher information, a constant factor 1/4 is not present. Here, for convenience, we have added the constant factor 1/4 in our definition, Eq. (\ref{eq:FI}), to compare with the versions commonly used in quantum Fisher information through evolution. With this prescription, the Fisher information ${\cal F}(\rho, H)$ and the variance ${\cal V}(\rho, H)$ of the observable $H$ are the same for any pure state.
\bibitem{note3} There is an equivalence between $[\rho, H]=0$ and ${\cal F}(\rho, H)=0$. We know that the vanishing of the commutator between $\rho$ and $H$ (i.e., $[\rho, H] = 0$) means that $\rho$ and $H$ have the same spectrum (not necessarily that they have the same eigenvalues). It then follows from Eq. (\ref{eq:QFI2}) that ${\cal F}(\rho, H) = 0$.
\end{thebibliography}

\end{document}